\documentclass[11pt,a4paper]{article}

\usepackage[utf8]{inputenc}
\usepackage[T1]{fontenc}
\usepackage{lmodern}
\usepackage{amsmath, amsfonts, amsthm, amssymb, mathtools}
\usepackage{setspace}
\usepackage{Tabbing}
\usepackage{booktabs}
\usepackage{lastpage}
\usepackage{extramarks}
\usepackage{chngpage}
\usepackage{soul, color}
\usepackage{graphicx, float, wrapfig}
\usepackage{verbatim}
\usepackage[ruled, vlined, linesnumbered]{algorithm2e}
\usepackage{url}

\newtheorem{theorem}{Theorem}
\newtheorem{lemma}{Lemma}

\newtheorem{corollary}{Corollary}

\theoremstyle{definition}
\newtheorem{definition}{Definition}

\theoremstyle{remark}
\newtheorem{claim}{Claim}

\numberwithin{equation}{section}
\numberwithin{figure}{section}

\usepackage{tikz}
\usetikzlibrary{trees, decorations}
\usetikzlibrary{arrows, automata}
\usetikzlibrary{matrix}
\usetikzlibrary{calc}

\usepackage[hmargin=2cm, vmargin=3.5cm]{geometry} % Shrink the margins.

\title{A Simple Condition for the Existence of Transversals}

\author{Arindam Biswas \thanks{Chennai Mathematical Institute, \protect\url{ari_b@cmi.ac.in}}}

\date{}

\begin{document}
\maketitle

\begin{abstract}
Hall's Theorem is a basic result in Combinatorics which states that the obvious necesssary condition for a finite family of sets to have a transversal is also sufficient. We present a sufficient (but not necessary) condition on the sizes of the sets in the family and the sizes of their intersections so that a transversal exists. Using this, we prove that in a bipartite graph $G$ (bipartition $\{A, B\}$), without 4-cycles, if $\deg(v) \geq \sqrt{2e|A|}$ for all $v \in A$, then $G$ has a matching of size $|A|$. 
\end{abstract}

%% Body %%
\section{Introduction and Preliminaries}
In this paper, we look at a sufficient (but not necessary) condition for a finite family of sets to have a transversal, which only involves the sizes of the sets and the sizes of their intersections.

We begin by recalling some theorems and definitions necessary for our presentation.

\begin{definition}[Transversal]
Let $\mathcal{F} = \{S_1, \ldots, S_n\}$ be a finite family of sets. A transversal for $\mathcal{F}$ is a tuple $(T, \phi)$, where $T \subseteq \bigcup \mathcal{F}$, and $\phi: T \to \mathcal{F}$ is a bijection such that $x \in \phi(x)$ for all $x \in T$.
\end{definition}

\begin{theorem}[Hall's Theorem]
Let $\mathcal{F} = \{S_1, \ldots, S_n\}$ be a finite family of sets. If for each $\mathcal{F'} \subseteq \mathcal{F}$, 
\[
    \left|\bigcup \mathcal{F'}\right| \geq |\mathcal{F'}|,
\]
then $\mathcal{F}$ has a transversal.
\end{theorem}

For an elegant proof of the above result, the reader may refer to~\cite{aigner2010}.

Consider a finite set of objects $S$ and a property $P$. Suppose we have a probability distribution on $S$. If an element $x \in S$ picked randomly according to the distribution has property $P$ with positive probability, then the set $S$ has atleast one element with property $P$. This is the basic idea behind the ``Probabilistic Method'', which has been used to produce elegant proofs of many combinatorial results.

In what follows, we formalise this idea.

\begin{definition}[Dependency Digraph]
Let $E_1, \ldots, E_n$ be events in a probability space. The \emph{dependency digraph} for these events is the graph $G = (V, E)$, where $V = \{1, \ldots, n\}$ and $E = \{(i, j) \mid 1 \leq i, j \leq n,\ i \neq j,\ E_i \textrm{ and } E_j \textrm{ are dependent events}\}$.
\end{definition}

\begin{lemma}[The Lovász Local Lemma]
Let $E_1, \ldots, E_n$ be events in a probability space and $D$ be their dependency digraph. Suppose there are numbers $0 \leq x_i < 1,\ 1 \leq i \leq n$ such that
\[
    P(E_i) \leq x_i \prod_{(i, j) \in E(D)} (1 - x_j), \textrm{ for all } 1 \leq i \leq n.
\]
Then
\[
    P\left(\bigwedge_{i = 1}^n E_i^c\right) \geq \prod_{i = 1}^n (1 - x_i).
\]
In particular, the probability that none of the events $E_i$ occur is positive.
\end{lemma}

\begin{corollary}
Let $E_1, \ldots, E_n$ be events in a probability space such that each event is dependent on at most $d$ other events and $P(E_i) \leq p, \textrm{ for all } 1 \leq i \leq n$. If
\[
ep(d + 1) \leq 1,
\]
then
\[
    P\left(\bigwedge_{i = 1}^n E_i^c\right) > 0.
\]
\end{corollary}

Proofs of the above lemma and its corollary can be found in~\cite{alon2008}.

Now suppose we want to prove that in a set of objects, there is an object with a certain property. If we can identify ``bad events'', which prevent a randomly picked object from having the property, then by the above corollary, it suffices to show that each ``bad event'', occurs with small enough probability and it does not depend on too many other ``bad events''.

In the next section, using the corollary, we obtain a simple sufficient condition for the existence of a transversal for a finite family of \emph{finite} sets.

\section{The Condition}
We are now ready to present our result.

\begin{theorem}
Let $\mathcal{F} = \{S_1, \ldots, S_n\}$ be a finite family of finite sets. Suppose there are numbers $l$ and $m$ such that $|S_i| \geq l$ and $|S_i \cap S_j| \leq m$, for all $1 \leq i, j \leq n$. If
\[
    \sqrt{em(2n - 3)} \leq l,
\]
then $\mathcal{F}$ has a transversal.
\end{theorem}
\begin{proof}
Consider a random n-tuple $(X_1, \ldots, X_n)$, where $X_i$ is a uniform random variable over $S_i$, for $1 \leq i \leq n$. For $1 \leq i, j \leq n$, denote by $E_{ij}$ the event $\{X_i = X_j\}$. The random tuple is a transversal when none of the events $E_{ij},\ 1 \leq i < j \leq n$ occur. This is exactly the event $\displaystyle\smashoperator{\bigwedge_{1\leq i < j \leq n}} E_{ij}^c$.

We have
\begin{align*}
    P(E_{ij}) &= \frac{|S_i \cap S_j|}{|S_i||S_j|}\\
              &\leq \frac{m}{l^2}.
\end{align*}
Each event $E_{ij}$ only depends on another event $E_{i'j'}$ if either $i = i'$ or $j = j'$. This can happen in $(n - 2) + (n - 2) = 2n - 4$ ways. Applying the Local lemma with $p = \frac{m}{l^2}$ and $d = 2n - 4$, we have
\[
P\left(\smashoperator[r]{\bigwedge_{1\leq i < j \leq n}} E_{ij}^c\right) > 0,
\]
if
\begin{align*}
    ep(d + 1) = e \frac{m}{l^2} (2n - 4 + 1) &\leq 1, \textrm{ i.e.}\\
    em(2n - 3) &\leq l^2, \textrm{i.e.}\\
    \sqrt{em(2n - 3)} &\leq l.
\end{align*}

Thus the event $\displaystyle\smashoperator{\bigwedge_{1\leq i < j \leq n}} E_{ij}^c$ occurs with positive probability, i.e.\ the family $\mathcal{F}$ has a transversal when $\sqrt{m(2n - 3)} \leq l$.
\end{proof}

\section{Matchings in Graphs}
By considering the set of neighbours of each vertex of a graph as an element of a family of sets, we can use the previous theorem to obtain the following result about the existence of matchings which saturate one of the vertex sets in a bipartition of a graph.

\begin{theorem}
Let $G = (V, E)$ be a bipartite graph with no 4-cycles. Let $\{A, B\}$ be a bipartition of $V$ and $n = |A|$. If $\deg(v) \geq \sqrt{2en}$ for all $v \in A$, then $G$ has a matching which saturates $A$.
\end{theorem}
\begin{proof}
For each $v \in A$, define
\[
    S_v = \{u \in B \mid \textrm{ u is a neighbour of v}\}.
\]

Since $\{A, B\}$ is a bipartition, all neighbours of vertices in $A$ are in $B$. Thus we have
\[
    |S_v| = \deg(v) \geq \sqrt{2en}.
\]

\begin{claim}[For all $u, v \in A,\ |S_u \cap S_v| \leq 1$]
Suppose $|S_u \cap S_v| \geq 2$. Then there are vertices $u', v' \in B$ such that $uu', uv', vu', vv'$ are edges in $G$. Thus $G$ contains the 4-cycle $uu'vv'u$. This contradicts our assumption that $G$ is 4-cycle-free.
\end{claim}
Now we apply our theorem from the previous section with $l = \sqrt{2en}$ and $m = 1$. Since
\[
\sqrt{em(2n - 3)} = \sqrt{e(2n - 3)} \leq \sqrt{2en} = l,
\]
the family $\mathcal{F} = \{S_v \mid v \in A\}$ has a transversal. Let $(B', \phi)$ be the transversal. Then
\[
M = \{\{v, \phi^{-1}(S_v)\} \mid v \in A\}
\]
is a matching which saturates $A$.
\end{proof}

\end{document}